%% file: ACC_with_IDS.tex
\documentclass[journal]{IEEEtran}

\AtBeginDocument{%
  \providecommand\BibTeX{{%
    \normalfont B\kern-0.5em{\scshape i\kern-0.25em b}\kern-0.8em\TeX}}}

\usepackage[utf8]{inputenc}
\usepackage{xcolor}
\usepackage{amsthm }
\usepackage{hyperref}
\usepackage{color}
\usepackage{listings}
\usepackage{acronym}
\usepackage[T1]{fontenc}
\usepackage{multirow}
\usepackage{colortbl}
\usepackage{todonotes}
\usepackage{url}
\usepackage[export]{adjustbox}
\usepackage{colortbl}
\usepackage{tcolorbox}
\usepackage{lipsum}
\usepackage{multicol}
\usepackage{float}
\usepackage{subcaption}
\usepackage{enumerate}
\usepackage{amsmath}
\usepackage{amssymb}

\usepackage{algorithm}
\usepackage{algpseudocode}

\definecolor{lightgray}{rgb}{.9,.9,.9}
\definecolor{darkgray}{rgb}{.4,.4,.4}
\definecolor{darkgreen}{rgb}{0, 0.39, 0.00}
\definecolor{Gray}{gray}{0.7}
\definecolor{codegreen}{rgb}{0,0.6,0}
\definecolor{codegray}{rgb}{0.5,0.5,0.5}
\definecolor{codepurple}{rgb}{0.58,0,0.82}
\definecolor{backcolour}{rgb}{0.95,0.95,0.92}

\lstdefinestyle{mystyle}{
    backgroundcolor=\color{backcolour},   
    commentstyle=\color{codegreen},
    keywordstyle=\color{magenta},
    numberstyle=\tiny\color{codegray},
    stringstyle=\color{codepurple},
    basicstyle=\ttfamily\footnotesize,
    breakatwhitespace=false,         
    breaklines=true,                 
    captionpos=b,                    
    keepspaces=true,                 
    numbers=left,                    
    numbersep=5pt,                  
    showspaces=false,                
    showstringspaces=false,
    showtabs=false,                  
    tabsize=2
}

\definecolor{codegreen}{rgb}{0,0.6,0}
\definecolor{codegray}{rgb}{0.5,0.5,0.5}
\definecolor{codepurple}{rgb}{0.58,0,0.82}
\definecolor{backcolour}{rgb}{0.95,0.95,0.92}
\definecolor{lightgray}{rgb}{.9,.9,.9}
\definecolor{darkgray}{rgb}{.4,.4,.4}
\definecolor{darkgreen}{rgb}{0, 0.39, 0.00}
\definecolor{Gray}{gray}{0.7}

\lstdefinestyle{mystyle}{
    backgroundcolor=\color{backcolour},   
    commentstyle=\color{codegreen},
    keywordstyle=\color{magenta},
    numberstyle=\tiny\color{codegray},
    stringstyle=\color{codepurple},
    basicstyle=\ttfamily\footnotesize,
    breakatwhitespace=false,         
    breaklines=true,                 
    captionpos=b,                    
    keepspaces=true,                 
    numbers=left,                    
    numbersep=5pt,                  
    showspaces=false,                
    showstringspaces=false,
    showtabs=false,                  
    tabsize=2https://www.overleaf.com/project
}

\newtheorem{theorem}{Theorem}[section] 
\newtheorem{lemma}[theorem]{Lemma}

\newtheorem{assumption}{Assumption}

\include{acronyms}

\begin{document}

\title{Extending Adaptive Cruise Control with Machine Learning Intrusion Detection Systems}

\author{Lotfi Ben Othmane,
        Yasaswini Konapalli,
        Naga Prudhvi Mareedu}


\maketitle

\begin{abstract}

\input{Sections/abstract}

\end{abstract}


\section{Introduction}
\input{Sections/Introduction}

\section{Overview of the Adaptive Cruise Control Model}\label{sec:ACCModel}
\input{Sections/ACCModel}

\section{Related work}~\label{sec:relworks}
\input{Sections/Relatwork}

\section{Limitation of Adaptive Cruise Control with Kalman Filter }~\label{sec:limitationACCKL}

\input{Sections/LimitationKFACC}

\section{Extending Adaptive Cruise Control with Intrusion Detection Systems}~\label{sec:IDSACC}

\input{Sections/IDSACC}

\section{Simulation of the proposed ACC-IDS model}~\label{sec:Simulation}
\input{Sections/Simulation}

\section{Conclusion}~\label{sec:Conclusions}
\input{Sections/Conclusion}

\section*{Acknowledgment}
During the preparation of this work the authors used ChatGPT to help developing the simulation code and to improve the composition of the paper.

\bibliographystyle{IEEEtran}
\bibliography{Bibliography}


\end{document}

%% file: acronyms.tex
\acrodef{ECU}{Electronic Control Unit}
\acrodef{CACC}{Cooperative Adaptive Cruise Control}
\acrodef{CC}{Cruise Control}
\acrodef{RSU}{Road Side Unit}
\acrodef{SC}{Service Center}
\acrodef{ITS}{Intelligent Transportation System}
\acrodef{OBU}{On-Board Unit}
\acrodef{OBD}{On-Board Diagnostics}
\acrodef{CAN}{Controller Area Network}
\acrodef{DoS}{Denial of Service}
\acrodef{AAA}{Automobile Association of America}
\acrodef{V2V}{Vehicle-to-Vehicle}
\acrodef{V2I}{Vehicle-to-Infrastructure}
\acrodef{PID}{Priority ID}
\acrodef{CIDS}{Clock-based IDS}
\acrodef{RLS}{Recursive Least Squares}
\acrodef{IDS}{Intrusion Detection System}
\acrodef{IVI}{In Vehicle Infotainment}
\acrodef{HMM}{Hidden Markov Model}
\acrodef{ML}{Machine Learning}
\acrodef{SAE}{Society of Automotive Engineers}
\acrodef{RPM}{Revolutions Per Minute}
\acrodef{CSMA/CD}{Carrier Sense Multiple Access with Collision Detection} 
\acrodef{ML}{Machine Learning}
\acrodef{NN}{Neural Networks}
\acrodef{TPM}{Tire Pressure Monitoring System}
\acrodef{GPS}{Global Positioning System}
\acrodef{DNN}{Deep Neural Network}
\acrodef{LSTM}{Long Short-Term Memory}
\acrodef{MSG}{Messages-Sequence Graph}
\acrodef{CPD}{Change Point Detection}
\acrodef{RNN}{Recurrent Neural Networks}
\acrodef{GID}{Generative Adversarial Nets based Intrusion  Detection System}
\acrodef{CNN}{Convolution Neural Networks}
\acrodef{GRU}{Gated Recurrent Unit}
\acrodef{GAN}{Generative Adversarial Nets}
\acrodef{OBD}{On-Board Diagnostics}
\acrodef{RT}{Real-Time}
\acrodef{ACC}{Adaptive Cruise Control }
\acrodef{DDoS}{Distribute Denial of Service}
\acrodef{CPS}{Cyber Physical System}
\acrodef{IDS}{Intrusion Detection System}
\acrodef{CAN}{Controller Area Network}
\acrodef{HMM}{Hodden Markov Model}
\acrodef{SAE}{Society of Automotive Engineers}
\acrodef{IPC}{Inter-Process Communication}
\acrodef{PID} {Proportional-Integral-Derivative}
\acrodef{ADAS}{Advanced Driver-Assistance System}
\acrodef{MPC}{Model Predictive Controller}
\acrodef{DOS}{Denial of Service}
\acrodef{IQR}{linear-quadratic optimal control}
\acrodef{LKA}{Lane Keeping Assist}
\acrodef{ACC-IDS}{Intrusion Detection Systems-Based Adaptive Cruise Control}
\acrodef{KF}{Kalman filter}
\acrodef{EKF}{Extended and Unscented Kalman Filters}
\acrodef{UKF}{Extended and Unscented Kalman Filters}
\acrodef{FDI}{False Data Injection}
\acrodef{CMDP}{Constrained Markov Decision Process}

%% file: Sections/abstract.tex
An \ac{ACC} system automatically adjusts the host vehicle’s speed to maintain a safe following distance from a lead vehicle. In typical implementations, a feedback controller (e.g., a \ac{PID} controller) computes the host vehicle’s acceleration using a target speed and a spacing error, defined as the difference between the measured inter-vehicle distance and a desired safe distance. \ac{ACC} is often assumed to be resilient to fault-injection attacks because a \ac{KF} can smooth noisy speed measurements. However, we show—through analytical proofs and simulation results—that a \ac{KF} can tolerate injected speed values only up to a bounded threshold. When injected values exceed this threshold, the filter can be driven off track, causing the \ac{ACC} controller to make unsafe acceleration decisions and potentially leading to collisions. Our main contribution is to augment the \ac{PID}-based controller with \ac{IDS} outputs, yielding \ac{ACC-IDS}. The \ac{ACC-IDS} controller is simple and implementable: a binary intrusion flag switches the control law to emergency braking. We prove that augmenting \ac{ACC} with an \ac{IDS}, under assumed detection-performance and latency constraints, can mitigate these attacks and help preserve \ac{ACC}’s collision-avoidance guarantees.

%% file: Sections/Introduction.tex
\acf{ACC} is a driver-assistance system designed to help vehicles maintain safe and efficient longitudinal control on the road. The system uses onboard sensors—such as radar, LiDAR, or cameras—to measure the distance to nearby vehicles and obstacles. Using this sensor measurements, \ac{ACC} regulates the host vehicle’s speed to match a driver-selected value while ensuring a safe following distance from leading vehicles. To achieve this, the controller considers inputs such as the desired speed, the vehicle’s actual speed, and the measured inter-vehicle distance. From these inputs, \ac{ACC} computes a safe distance based on vehicle dynamics and adjusts the throttle or brakes to reduce tracking error.

Cyberattacks on modern vehicles have progressed from controlled laboratory demonstrations to real-world incidents, driven largely by the increasing number of software and communication vulnerabilities in automotive systems~\cite{lee2022evaluation,UpstreamAuto2020}. These vulnerabilities arise in components such as electronic control units, in-vehicle networks, and infotainment systems, providing multiple attack surfaces for adversaries~\cite{Othmane2015}. The consequences range from vehicle theft to the induction of hazardous driving conditions, including traffic accidents~\cite{6728305}. In human-in-the-loop scenarios, drivers are expected to recognize abnormal or unsafe behavior and intervene when necessary, underscoring the importance of understanding how automated vehicle systems interact with human operators.

As autonomous and semi-autonomous vehicles become more widespread, systems like \ac{ACC} are essential for automatically regulating vehicle speed and headway~\cite{rajamani2012vehicle}. Conventional \acf{CC} maintains a constant driver-selected speed by adjusting the throttle or brakes based on speed error. \ac{ACC} extends \ac{CC} by additionally incorporating the distance to a lead vehicle, as illustrated in Figure~\ref{fig:ACCoverview}. To fulfill these objectives, \ac{ACC} typically employs a hierarchical control architecture: an upper-level controller computes the desired longitudinal acceleration to maintain safe spacing or constant speed, and a lower-level actuator controller converts this acceleration command into throttle or brake signals~\cite{rajamani2012vehicle}.

\begin{figure}
\includegraphics[width=1.0\linewidth]{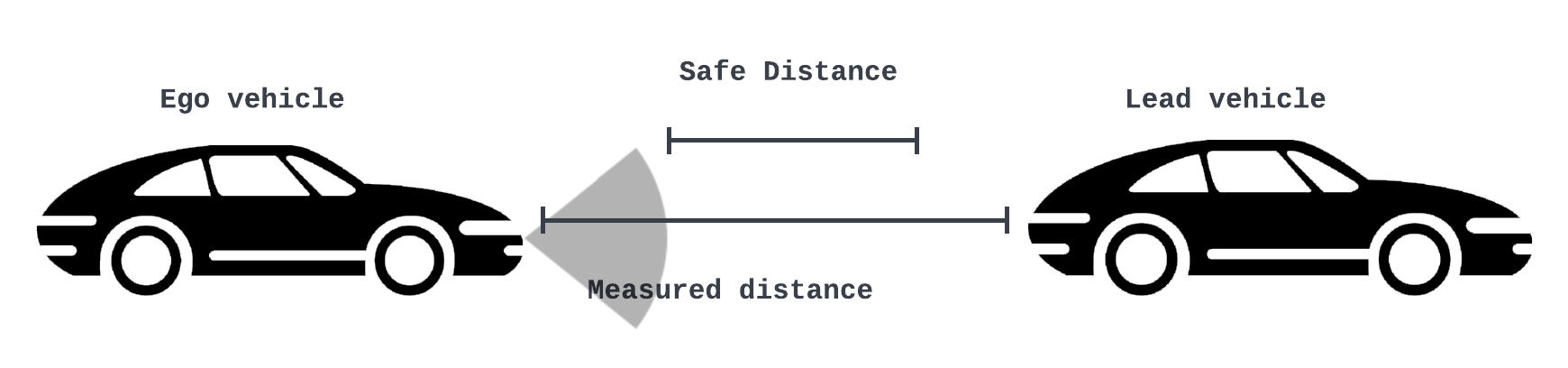}
\caption{The \ac{ACC} system uses radar or LiDAR to measure the distance to the lead vehicle and compares it to a safe following distance.}
\label{fig:ACCoverview}
\end{figure}

Adversaries often probe these systems to learn their behavior and identify weaknesses. For example, to manipulate the vehicle’s perceived speed, an attacker may repeatedly inject falsified speed measurements, observe the controller’s response, and adapt the forged signals over time. Controllers commonly use estimators—such as the \acf{KF}~\cite{9000}—to reject abrupt or implausible measurement jumps. While effective against random noise, such estimators are less robust to strategically crafted, gradual spoofing attacks~\cite{9076852}.

The upper-level \ac{ACC} controller computes acceleration commands from an error signal that captures the deviation between the desired and actual spacing and/or speed. This logic is commonly implemented using a \ac{PID} controller~\cite{rajamani2012vehicle} and a \ac{KF} to address speed sensor noise and compute vehicle acceleration. A \ac{KF} is optimal under its modeling assumptions and is known to be effective at suppressing random fluctuations in measured data. A spoofing attacker, however, can violate these assumptions by injecting biased or adaptively crafted measurements that systematically distort the estimate. For example, in~\cite{JOS2024} we simulated speed and RPM spoofing attacks in the CARLA simulator~\cite{Dosovitskiy17}. The results show that spoofing the host vehicle’s speed to a fixed value of 90km/h consistently resulted in a collision, whereas spoofing to lower target speeds, such as 60km/h, did not lead to accidents.

When an adversary is able to manipulate the measurement channel in a sustained or adaptive manner, estimation-based techniques alone are insufficient to guarantee safety. In contrast, an intrusion detection system operates independently of the state estimator and does not attempt to reconstruct the true vehicle state. Instead, it focuses on identifying deviations from expected behavioral patterns in sensor readings and in-vehicle communication. Machine-learning-based IDS approaches are particularly effective in this setting because they can capture complex temporal relationships and contextual correlations that characterize normal system behavior, enabling the detection of gradual and stealthy spoofing attacks that may evade model-based or threshold-driven filters. Crucially, the output of an IDS can be directly integrated into the control logic, allowing the controller to transition to a conservative, fail-safe mode when an attack is detected, thereby preserving safety even when the reliability of state estimation is compromised.

The contributions of this paper are:
\begin{enumerate}
\item Demonstrate the limitations of \ac{ACC} systems that rely on \ac{KF}-based estimation in preventing collisions under cyberattacks.
\item Evaluate the effectiveness of augmenting \ac{ACC} with a machine-learning-based \ac{IDS} to mitigate such attacks, both analytically and through simulation.
\end{enumerate}

The remainder of this paper is organized as follows. Section~\ref{sec:ACCModel} introduces the \ac{ACC} control framework. Section~\ref{sec:relworks} reviews related work. Section~\ref{sec:limitationACCKL} presents the limitations of \ac{KF}-based \ac{ACC} under spoofing attacks. Section~\ref{sec:IDSACC} evaluates machine-learning-based \ac{IDS} integration for attack mitigation. Section~\ref{sec:Simulation} discusses the results of simulating \ac{ACC} with \ac{KF}, fault-injection attacks, and with \ac{IDS}. Section~\ref{sec:Conclusions} concludes the paper.

%% file: Sections/ACCModel.tex
This section describes the longitudinal control structure of automotive \ac{CC} and \ac{ACC} systems following the formulation in~\cite{rajamani2012vehicle}. Conventional \ac{CC} maintains a driver-selected reference speed by regulating the throttle and brake based on speed error. \ac{ACC} extends this by also considering the distance to a lead vehicle, as illustrated in Figure~\ref{fig:ACCoverview}.

\ac{ACC} typically employs a hierarchical control architecture (see Figure~\ref{fig:ACC_controllers}). The upper-level controller computes the desired longitudinal acceleration required either to track the reference speed (in cruise mode) or maintain a safe following distance (in following mode).  
The lower-level controller translates the desired acceleration into physical actuator commands for throttle and braking~\cite{rajamani2012vehicle}.

\begin{figure}[bt]
    \centering
    \includegraphics[width=0.6\linewidth]{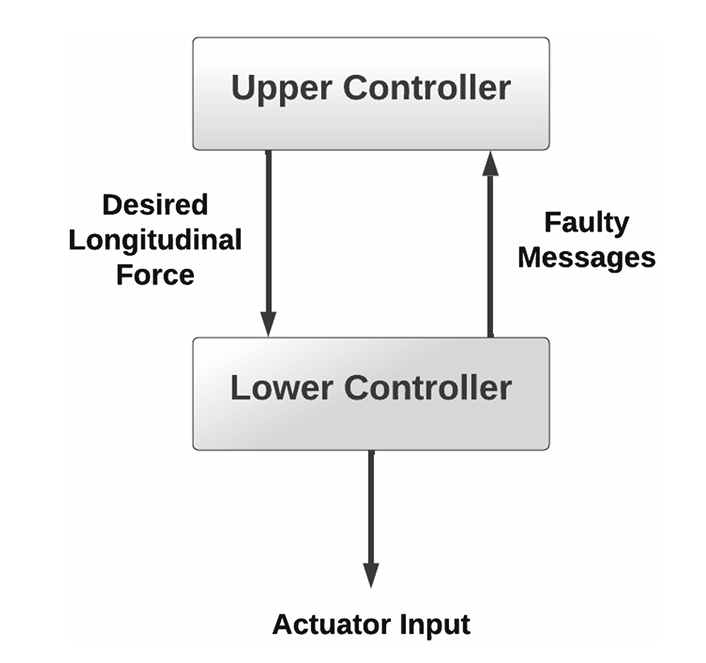}
    \caption{Upper- and lower-level controllers in the \acf{ACC} architecture~\cite{JOS2024}.}
    \label{fig:ACC_controllers}
\end{figure}

\paragraph{Safe Following Distance}

The desired inter-vehicle distance is computed as~\cite[p.207]{greenbook}:\footnote{This is a simplified formula as it groups the constants in the original one.}

\begin{equation}\label{eq:safedistance}
d_{safe}(t) = h\, v_h(t)
+ \frac{v_h(t)^2}{2a}
\end{equation}

The first term of Equation~\ref{eq:safedistance} models the reaction-distance traveled during human braking delay and the second term models the physical braking distance. The constant $h$ is the headway time (typically $\sim$1.5--2.5\,s), $v_h(t)$ is the host velocity in m/s, and $a$ is the assumed comfortable deceleration rate (typically $\sim$3.4\,m/s$^2$), and 2 accounts for kinematic constants rolled together~\cite{greenbook}. 

\begin{table}[bt]
    \centering
    \caption{Representative PID controller gains for \ac{ACC} (adapted from~\cite{rajamani2012vehicle}).}
    \label{tab:PIDConstants}
    \begin{tabular}{|p{1.7in}|p{0.35in}|p{0.35in}|p{0.35in}|}
        \hline
        \rowcolor{gray!18}
        Context & $K_p$ & $K_i$ & $K_d$ \\ \hline
        Mid-size passenger car $(m\approx1500\,\text{kg})$ 
        & 0.2--1.0 & 0.05--0.2 & 0.1--0.5 \\ \hline
        Speed-control mode (gentler dynamics) 
        & 0.2 & 0.1 & 0 \\ \hline
        Spacing-control mode (more responsive) 
        & 0.5 & 0.05 & 0.2 \\ \hline
    \end{tabular}
\end{table}

\paragraph{Tracking Errors}

Two tracking errors are considered:
\begin{align}
e_v(t) &= v_{ref}(t) - \hat v_h(t), \label{eq:speederror}\\[4pt]
e_d(t) &= d(t) - d_{safe}(t), \label{eq:spaceingerror}
\end{align}
representing speed error (Eq.~\ref{eq:speederror}) and spacing error (Eq.~\ref{eq:spaceingerror}), respectively. Note that Equation~\ref{eq:speederror} uses the host speed estimated using the Kalman filter ($\hat v_h$).  

The controller selects the relevant error type based on whether a lead vehicle is detected or not as specified in the following equations:
\begin{equation}
e(t)=
\begin{cases}
e_d(t), & \text{if a lead vehicle is present},\\[2pt]
e_v(t), & \text{otherwise}.
\end{cases}
\label{eq:errorswitch}
\end{equation}

\paragraph{Control Objectives}

The \ac{ACC} seeks to minimize the errors as specified in the following equations.

\begin{align}
\lim_{t\to\infty} e_v(t) &= 0, \label{eq:CruiseModeOptimal}\\
\lim_{t\to\infty} e_d(t) &= 0,\qquad d(t)\ge d_{safe}(t), \label{eq:SpaceModeOptimal}
\end{align}

Equations~\ref{eq:CruiseModeOptimal} and~\ref{eq:SpaceModeOptimal} correspond to speed tracking and safe following distance errors, respectively. It achieves this by adjusting the vehicle acceleration, $u(t)$, using controllers such as the \ac{PID} controller.

\paragraph{PID Controller}

A PID controller is commonly used for the upper-level control law:
\begin{equation}
u(t)
= K_p e(t)
+ K_i \int_0^t e(\tau)\,d\tau
+ K_d \frac{d e(t)}{dt},
\label{eq:PID}
\end{equation}
where $u(t)$ denotes the desired longitudinal acceleration.  
Typical ranges for PID gains are shown in Table~\ref{tab:PIDConstants}; these values depend on vehicle mass, actuator dynamics, headway selection, and comfort/safety trade-offs.

\begin{figure}[h]
    \centering
    \includegraphics[width=1\linewidth]{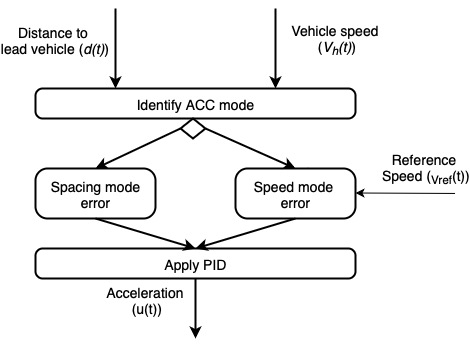}
    \caption{Block diagram of the \acf{ACC} system.}
    \label{fig:ACCBlockDiagram}
\end{figure}

Let's put things together. Figure~\ref{fig:ACCBlockDiagram} illustrates the overall signal flow in \ac{ACC}.  
The controller receives as inputs the measured distance to the lead vehicle $d(t)$, the estimated host speed $v_h(t)$, and the driver-selected reference speed $v_{ref}(t)$, and produces a longitudinal acceleration command $u(t)$.

\ac{ACC} alternates between speed regulation and spacing control using a PID-based upper-level controller, with gains tuned to ensure stability, safety, and passenger comfort~\cite{ploeg2011design}.

%% file: Sections/Relatwork.tex
Jedh \emph{et al.}~\cite{JOS2024} employed the CARLA simulator~\cite{Dosovitskiy17} to model a host vehicle equipped with adaptive cruise control (\ac{ACC}) following a lead vehicle and evaluated crash-avoidance behavior under normal and adversarial conditions. The authors introduced speed- and RPM-spoofing attacks in which falsified speed measurements are propagated over the vehicle’s Controller Area Network (\ac{CAN}) bus. Their results showed that speed spoofing can mislead the \ac{ACC} controller into accelerating when deceleration is required, thereby increasing the risk of rear-end collisions. To mitigate this risk, the \ac{ACC} system was augmented with a real-time, machine-learning-based intrusion detection system (\ac{IDS}) that triggers emergency braking when an imminent collision is detected, demonstrating that IDS-assisted braking can significantly reduce the safety impact of spoofing attacks.

Sardesai \emph{et al.}~\cite{8656970} investigated the relationship between cybersecurity and safety in \ac{CACC} systems by analyzing adversarial manipulation of \ac{V2V} communication in a highway merging scenario. The authors developed an abstract \ac{CACC} model in MATLAB and simulated passive attack strategies, including message injection and denial-of-service–like disruptions that corrupt transmitted speed and position data. System behavior was evaluated using safety-critical indicators such as collision occurrence, false alarms, and missed hazard detection. The results showed that even a limited fraction of corrupted \ac{V2V} messages can lead to incorrect control actions and a non-negligible probability of collisions, emphasizing the importance of secure and reliable inter-vehicle communication.

Guo \emph{et al.}~\cite{GGSQ2020} proposed a multi-objective \ac{ACC} strategy based on \ac{MPC} that simultaneously considers safety, car-following performance, ride comfort, and fuel efficiency. The approach models longitudinal vehicle dynamics using higher-order kinematic states, including inter-vehicle distance, relative velocity, acceleration, and jerk, and introduces constraint softening to address infeasible operating conditions. Simulation results using MATLAB/Simulink and CarSim demonstrated improved robustness, tracking accuracy, and stability compared to conventional \ac{MPC}-based \ac{ACC} controllers.

Bersani \emph{et al.}~\cite{8804527} addressed vehicle state estimation for advanced driver-assistance systems and autonomous driving by estimating key motion states such as longitudinal and lateral velocities, heading angle, and absolute position. The authors developed adaptive \ac{EKF} and \ac{UKF} formulations based on a kinematic single-track vehicle model, reducing dependence on uncertain vehicle parameters while preserving real-time feasibility. Sensor fusion was achieved using inertial sensors, wheel-speed encoders, steering-angle sensors, and dual GPS receivers operating at different sampling rates. Experimental results from urban driving scenarios confirmed that both \ac{EKF} and \ac{UKF} provide accurate and stable state estimates suitable for real-time vehicle control.

Wu \emph{et al.}~\cite{8759909} proposed a \ac{CACC} strategy that maintains longitudinal control performance during extended communication loss by estimating the preceding vehicle’s acceleration instead of reverting to conventional \ac{ACC}. The method incorporates an adaptive Kalman filter within a predecessor–following \ac{CACC} architecture, enabling reconstruction of missing feedforward acceleration signals when \ac{V2V} communication is unavailable. Compared to traditional Kalman filters based on the Singer model, the proposed approach adaptively updates the acceleration mean using past estimates, reducing control oscillations and improving estimation accuracy. Simulation and experimental validation on mobile robot platoons demonstrated significant reductions in mean and RMS inter-vehicle distance errors.

Sargolzaei \emph{et al.}~\cite{8894512} presented a secure control framework for networked control systems that jointly detects and mitigates \ac{FDI} attacks in real time. The framework integrates a Kalman filter–based state observer with a neural network whose parameters are updated using an extended Kalman filter, enabling simultaneous estimation of system states and malicious inputs under process and measurement noise. By embedding attack estimates directly into a resilient control law, the approach compensates for uncertainties and unknown inputs without requiring controller reconfiguration. Simulation results on a multi-area power system demonstrated improved robustness and detection accuracy compared to conventional observer-based techniques.

Recent research has also focused on enhancing decision-making robustness under explicit safety constraints. Zhao \emph{et al.} proposed a safe reinforcement learning framework for autonomous highway driving by formulating the driving problem as a constrained \ac{CMDP}. To balance safety and learning efficiency, the authors introduced a Replay Buffer Constrained Policy Optimization algorithm that employs importance sampling to stabilize training. Simulation results demonstrated faster convergence, improved robustness to environmental uncertainties, and zero-collision performance in high-risk highway scenarios, motivating the integration of safety-aware mechanisms in autonomous driving systems such as IDS-enhanced \ac{ACC}~\cite{rad2020experimental2}.

Zhang \emph{et al.}~\cite{ZYZQ2025} proposed a machine-learning-based framework for detecting and mitigating speed fault-injection, \ac{DoS}, and spoofing attacks in \ac{ACC} systems. Synthetic datasets were generated to model various attack scenarios, and multiple classical machine learning algorithms, including support vector machines, were evaluated for attack detection. System resilience was assessed by verifying preservation of the \ac{ACC} time headway between host and lead vehicles during and after attacks, demonstrating the effectiveness of data-driven detection approaches for improving \ac{ACC} safety under adversarial conditions.

%% file: Sections/LimitationKFACC.tex
\begin{table}[bt]
    \centering
    \caption{Kalman filter notation for host-speed estimation.}
    \label{tab:Notation}
    \begin{tabular}{|p{0.7in}|p{2.4in}|}\hline
    \rowcolor{gray!18}
    Variable & Description \\\hline
             $\Delta t$ & Sampling period (instantiated to 1)\\\hline
       $v_h(t)$ & Effective host-vehicle speed at time $t$ (in m/s) \\\hline
       $u(t)$ & Host-vehicle acceleration at time $t$\\\hline
        $z(t)$ & Measured (noisy) speed at time $t$\\\hline
        $\hat v_h(t|t)$   & Updated (posterior) host-speed estimate at time $t$\\\hline
        $\hat v_h(t+ \Delta t|t)$ & Predicted host speed at time $t+\Delta t$ using measurements up to $t$\\\hline
        $\hat v_h(t+ \Delta t|t+ \Delta t)$ & Updated estimate at time $t+ \Delta t$ after incorporating the new measurement\\\hline 
        $n(t)$& Process noise at time $t$\\ \hline
        $Q$& Process-noise variance\\ \hline
        $m(t)$& Measurement noise at time $t$\\\hline
        $R$& Measurement-noise variance\\\hline
     $v_l(t)$ & Lead-vehicle speed at time $t$, assumed constant over one sampling step (in m/s)\\\hline
     $d(t)$ & Gap distance between the host and the lead vehicle at time $t$\\\hline
     $a$ & Comfortable braking acceleration \\\hline
     $v_h^*(t)$ & Simplified notation for $v_h(t+1)$\\\hline
    $\hat v_h^{-}(t)$ & Simplified notation for $\hat v_h(t+ \Delta t \mid t)$\\\hline
    $\hat v_h^{+}(t)$ & Simplified notation for $\hat v_h(t+ \Delta t \mid t+ \Delta t)$\\\hline
    \end{tabular}
\end{table}

The goal of this section is to characterize how combinations of (i) the measured host-vehicle speed after application of the \ac{KF} and (ii) the lead-vehicle speed can drive the \ac{ACC} system into an unsafe state, potentially leading to a collision. We decompose the analysis into two subproblems: (1) identifying the host-vehicle speed threshold $v_{\text{thr}}$ at which the \ac{ACC} safe-distance constraint is violated, and (2) deriving a condition on the \emph{measured} host-vehicle speed (in the presence of a \ac{KF}) under which the estimated speed exceeds $v_{\text{thr}}$. We first provide an overview of the \ac{KF} and then address these two subproblems.

Table~\ref{tab:Notation} summarizes the notation used throughout this section.

We make the following assumptions:

\begin{assumption}\label{as:1}
We consider a simple longitudinal model in discrete time with sampling period $\Delta t$.
\end{assumption}

\begin{assumption}\label{as:2}
The host-vehicle longitudinal dynamics are given by
\[
v_h(t+\Delta t) = v_h(t) + u(t)\,\Delta t,
\]
where $u(t)$ denotes the longitudinal acceleration.
\end{assumption}

\subsection{Prediction of Host-Vehicle Speed Using the Kalman Filter}
\label{sec:KF-prediction}

To estimate the host-vehicle speed $v_h(t)$, we use a standard discrete-time Kalman filter based on a constant-velocity motion model. The state is the host speed, i.e., $x(t)=v_h(t)$, and the measurement $z(t)$ is a noisy speed observation obtained, for example, from wheel-speed sensors, GPS, or a virtual sensor. The \ac{KF} model~\cite{9000} is
\begin{align}
v_h(t+ \Delta t) &= v_h(t) + n(t), \\
z(t) &= v_h(t) + m(t),
\end{align}
where $n(t)$ and $m(t)$ denote the process and measurement noise, respectively, and are typically modeled as zero-mean Gaussian random variables.

\paragraph{Step 1: Prediction to time $t+ \Delta t$}

Under the constant-velocity assumption, the one-step-ahead prediction equals the corrected estimate at the current step:
\begin{equation}\label{eq:KLVt+1}
\hat v_h(t+1 \mid t) = \hat v_h(t \mid t).
\end{equation}

The covariance propagates as
\begin{equation}
P(t+ \Delta t \mid t) = P(t \mid t) + Q,
\end{equation}
where $Q$ is the process-noise variance that captures uncertainty due to unmodeled speed variations.

\paragraph{Step 2: Update using the measurement {$z(t+ \Delta t)$}}

Given $z(t+ \Delta t)$, the Kalman gain is
\begin{equation}
K_{t+1} = \frac{P(t+ \Delta t \mid t)}{P(t+ \Delta t \mid t) + R},
\end{equation}
where $R$ denotes the measurement-noise variance.

The updated estimate is
\begin{equation}
\hat v_h(t+ \Delta t \mid t+ \Delta t)
=
\hat v_h(t+ \Delta t \mid t)
+
K_{t+1}\big[z(t+ \Delta t) - \hat v_h(t+ \Delta t \mid t)\big],
\end{equation}
and the associated covariance becomes
\begin{equation}
P(t+ \Delta t \mid t+ \Delta t) =
\big(1 - K_{t+1}\big)P(t+ \Delta t \mid t).
\end{equation}

This update smooths noisy speed measurements and provides a robust estimate of the host vehicle’s longitudinal speed for use in the \ac{ACC} controller.

\subsection{Identifying the Host-Vehicle Speed Threshold That Violates ACC}

This subsection proves the first lemma, which specifies the host-vehicle speed threshold that violates \ac{ACC}. To do so, we first state the following assumption.

\begin{assumption}\label{as:3}
The inter-vehicle gap after one sampling step $\Delta t$ is
\[
d(t+ \Delta t) = d(t) + \big(v_l(t+ \Delta t) - v_{h}(t+ \Delta t)\big)\, \Delta t.
\]
Note that we compute distance in meters and we use the speed at the end of the sampling period.
\end{assumption}

In practice, \ac{ACC} systems use radar or LiDAR to measure the range to the lead vehicle. Assumption~\ref{as:3} captures the one-step-ahead range evolution using the host vehicle’s expected speed at the next time step, computed under Assumption~\ref{as:2}. Accordingly, in predicting the inter-vehicle distance at $t+\Delta t$, we use the host vehicle's current effective (attack- and noise-free) speed and acceleration, rather than the potentially compromised or noisy speed measurement.

Recall that the desired safe distance is given by Eq.~\ref{eq:safedistance}:
\[
d_{\text{safe}}(t) = h\, v_h(t) + \frac{v_h(t)^2}{2a}.
\]

We distinguish between the effective host speed $v_h(t)$, which governs the true distance evolution, and the KF-updated speed $v_h^+(t)$ used by ACC to compute $d_{\text{safe}}$; attacks act by biasing $v_h^+$ while the true dynamics follow $v_h$.

\begin{lemma}[ACC speed violation at next-step gap] \label{lm:AccSpVilation}
A host vehicle maintains a safe distance to the lead vehicle, as formulated in Eq.~\ref{eq:safedistance}, if its speed $v_h^+(t)$ is less than or equal to the threshold
\[
\begin{split}
&v_{\text{thr}}(t+\Delta t)
= a\Bigg(-h\\
&+ \sqrt{h^2
+ \frac{2}{a}\big(d(t)+(v_\ell(t+\Delta t) - v_h(t+\Delta t))\Delta t\big)}\Bigg).
\end{split}
\]
\end{lemma}

Note that the safe distance is evaluated at the updated speed $v_{h}^+(t)$.

\noindent{\bf Proof.} 
We seek conditions at time $t+\Delta t$ under which the gap falls below the safe distance (i.e., a violation occurs):
\begin{equation}\label{eq.17}
d(t+\Delta t) < d_{\text{safe}}(t+\Delta t).
\end{equation}
Under Assumption~3, the true (effective) distance evolution over one step is
\begin{equation}\label{eq:gap_update_eff}
d(t+\Delta t)= d(t) + \big(v_\ell(t+\Delta t)-v_h(t+\Delta t)\big)\Delta t.
\end{equation}
The ACC module computes the safe distance using the KF-updated host speed $v_h^+(t)$, i.e.,

\begin{equation}\label{eq:safe_eval_updated}
d_{\text{safe}}(t+\Delta t)= hv_h^+(t)+\frac{(v_h^+(t))^2}{2a}.
\end{equation}

Substituting \eqref{eq:gap_update_eff} and \eqref{eq:safe_eval_updated} into \eqref{eq.17} yields
\begin{equation}\label{eq.18}
\frac{1}{2a}(v_h^+(t))^2 + h\,v_h^+(t)
-\Big(d(t)+v_\ell(t+\Delta t)\Delta t - v_h(t+\Delta t) \Delta t\Big) > 0.
\end{equation}

Define

$g(v)= p v^2 + b v + c,$

with
$p=\frac{1}{2a} \ (>0),\qquad$

$b=h\qquad$

$c=-\Big(d(t)+(v_\ell(t+\Delta t) - v_h(t+\Delta t))\Delta t\Big).$

Then \eqref{eq.18} is equivalent to $g(v_h^+(t))>0$.
The roots of $g(v)=0$ are

\[
v = \frac{-b\pm \sqrt{b^2-4pc}}{2p}.
\]
Since $p=\frac{1}{2a}$, we have $2p=\frac{1}{a}$ and thus $\frac{1}{2p}=a$.
Moreover, assume the predicted next-step gap remains positive, i.e.,
\[
d(t)+\big(v_\ell(t+\Delta t)-v_h(t+\Delta t)\big)\Delta t > 0,
\]
which implies $c<0$.
Because $p>0$, $g(v)\to+\infty$ as $v\to\infty$, and since $g(0)=c<0$, the equation $g(v)=0$ has exactly one positive root.
Therefore, the (unique) positive root is

\begin{equation}\label{eq:thr_pos}
\begin{split}
&v_{\text{thr}}(t+\Delta t)
= a\Bigg(-h\\
&+ \sqrt{h^2
+ \frac{2}{a}\big(d(t)+(v_\ell(t+\Delta t) - v_h(t+\Delta t)\Delta t\big)}\Bigg).
\end{split}
\end{equation}

while the other root is negative and thus not physically meaningful.

Because $p>0, g(v)\to+\infty$ as $v\to\infty$, and $g(0)=c$. Under normal operating conditions with a positive gap and bounded speeds, $g(v)$ crosses zero exactly once at the positive root; consequently,
$g(v_h^+(t))>0 \quad \Longleftrightarrow \quad v_h^+(t) > v_{\text{thr}}(t+\Delta t)$.
Thus,
$v_h^+(t) > v_{\text{thr}}(t+\Delta t)
\quad \Longrightarrow \quad
d(t+\Delta t) < d_{\text{safe}}(t+\Delta t)$,
which completes the proof. \qed

That is, assume that the current distance is safe: $d(t) \ge d_{\text{safe}}(t)$. If the host speed exceeds the threshold in the next step ($v_h^+(t) > v_{\text{thr}}(t+\Delta t)$), then the predicted next step is unsafe ($d(t+\Delta t) < d_{\text{safe}}(t+\Delta t)$).

We note that a violation of the speed threshold does not necessarily imply a collision between the host and lead vehicles. However, a sequence of violations may progressively reduce the inter-vehicle gap, potentially culminating in a crash.

\subsection{Measured Host-Vehicle Speed Values Exceeding the Threshold in the Presence of a Kalman Filter}

We next study how Kalman filtering can cause the \emph{estimated} host speed to exceed $v_{\text{thr}}$ even when the true speed is lower, for example, due to spoofed measurements.

\begin{lemma}[ACC speed violation with a Kalman filter] \label{lm:AccMSpVilation}
Under a Kalman filter, the updated host-speed estimate exceeds the safety threshold $v_{\text{thr}}(t+\Delta t)$ (and may thus violate the safe-distance constraint) if the measured speed satisfies
\[
z(t+\Delta t) >
\frac{v_{\text{thr}}(t+\Delta t) - (1 - K_{t+1})\hat v_h(t \mid t)}{K_{t+1}}.
\]
\end{lemma}

\noindent{\bf Proof sketch.}
Consider a 1D Kalman filter for the host speed $v_h(t)$ with measurement model
\begin{equation}
z(t+\Delta t) = v_h(t+\Delta t) + m(t+\Delta t),
\end{equation}
where $z(t+\Delta t)$ is the measured (possibly spoofed) speed and $m(t+\Delta t)$ denotes measurement noise.

The prediction and update equations are
\begin{align}
\hat v_h^{-}(t) &= \hat v_h(t+\Delta t \mid t), \\[4pt]
\hat v_h^{+}(t) &= \hat v_h(t+\Delta t \mid t+\Delta t)\\
&= \hat v_h^{-}(t) + K_{t+1}\big(z(t+\Delta t) - \hat v_h^{-}(t)\big),
\end{align}
where $\hat v_h^{-}(t)$ is the predicted speed, $\hat v_h^{+}(t)$ is the updated speed, and $K_{t+1}$ is the Kalman gain at time $t+\Delta t$, with $K_{t+1}>0$. (Also $K_{t+1}<1$)

We seek a condition on $z(t+\Delta t)$ such that $\hat v_h^{+}(t) > v_{\text{thr}}(t+\Delta t)$. Substituting the update equation and rearranging terms yields
\begin{equation}
z(t+\Delta t) >
\frac{v_{\text{thr}}(t+\Delta t) - (1 - K_{t+1})\hat v_h^{-}(t)}{K_{t+1}}.
\end{equation}
Using $\hat v_h^{-}(t)=\hat v_h(t \mid t)$ gives the stated result.

This expression specifies the minimum (possibly spoofed) measurement required to push the Kalman-filtered estimate above the threshold $v_{\text{thr}}(t+\Delta t)$. Recall that the \ac{KF} model uses a constant-velocity assumption with process noise.

\subsubsection*{PID Control in Cruise and Following Modes}

For completeness, we recall the PID control laws for the two operating modes of \ac{ACC}. 
In \emph{cruise mode}, the controller regulates speed:
\begin{equation}
\begin{split}
u(t) = K_p^v \big(v_{\text{ref}} - \hat v_h(t)\big)\\
&+ K_i^v \int_0^t \big(v_{\text{ref}} - \hat v_h(t)\big)\, dt\\
&+ K_d^v \big(\dot{v}_{\text{ref}} - \dot{v}_h(t)\big),
\end{split}
\label{eq:PIDCrusieMode}
\end{equation}
while in \emph{following mode}, it regulates the inter-vehicle distance:
\begin{equation}
\begin{split}
u(t) = K_p^d \big(d(t) - d_{\text{ref}}(t)\big)\\
&+ K_i^d \int_0^t \big(d(t) - d_{\text{ref}}(t)\big)\, dt\\
&+ K_d^d \big(\dot{d}(t) - \dot{d}_{\text{ref}}(t)\big).
\end{split}
\label{eq:PIDFollwoingmode}
\end{equation}
where $d_{\text{ref}}(t)$ denotes the reference distance between the host and lead vehicles (e.g., the safe distance).

Note that the controllers in Eq.~\ref{eq:PIDCrusieMode} and Eq.~\ref{eq:PIDFollwoingmode} are special cases of the PID controller in Eq.~\ref{eq:PID}.

%% file: Sections/IDSACC.tex
We propose augmenting the in-vehicle network with an \ac{IDS} to detect attacks that target the \ac{ACC} pipeline. An \ac{IDS} monitors \ac{CAN} communication, identifies anomalies indicative of compromised components (e.g., spoofed speed signals or manipulated controller outputs), and raises an alert when malicious activity is detected~\cite{mubarek2023}. Numerous \ac{IDS} solutions have been proposed to mitigate cyberattacks targeting in-vehicle networks and automotive control systems~\cite{Lee2017OTIDS,Tariq2020CANADF,Hossain2020LSTMCanBus,Amato2021CANBusDL,9490207,8640808,8688625}. In the following, we discuss integrating an \ac{IDS} with \ac{ACC} and assess the efficacy of this approach.

\subsection{Augmenting ACC with IDS (ACC-IDS)}

Figure~\ref{fig:BlockDiagramACCWithIDS} illustrates the proposed \ac{ACC-IDS} architecture. As in conventional \ac{ACC}, the controller receives the distance to the lead vehicle $d(t)$, the host-vehicle speed $v_h(t)$, and the reference speed $v_{\text{ref}}(t)$ at each time step. We extend this design by incorporating an additional signal from the \ac{IDS}: a binary intrusion indicator. When an intrusion is detected, the \ac{ACC-IDS} overrides the nominal control action and enforces braking.

We model the \ac{IDS} output as $S(t)\in\{0,1\}$, where $S(t)=1$ denotes attack detection and $S(t)=0$ otherwise. When $S(t)=1$, the controller should disregard the current sensor measurement because it may be spoofed. A naive countermeasure is to reuse the last known valid speed measurement; equivalently, one may replace the current sensor reading $z(t+\Delta t)$ in the Kalman filter update with the prediction $\hat v_h(t+\Delta t\mid t)$. However, this strategy is fragile in adversarial settings, where spoofed measurements may be injected at a higher rate than genuine ones. Consequently, relying solely on stale measurements or prediction substitution may be insufficient to withstand sustained spoofing attacks.

Instead, we modify the cruise-mode controller in Eq.~\ref{eq:PIDIDSCruise} so that braking is enforced whenever the \ac{IDS} signals an intrusion and we apply the braking acceleration and use the last known safe speed--ignoring the recent speed value associated with the injection. This yields the following control law.

\begin{figure}[tb]
\centering
\includegraphics[width=0.50\textwidth]{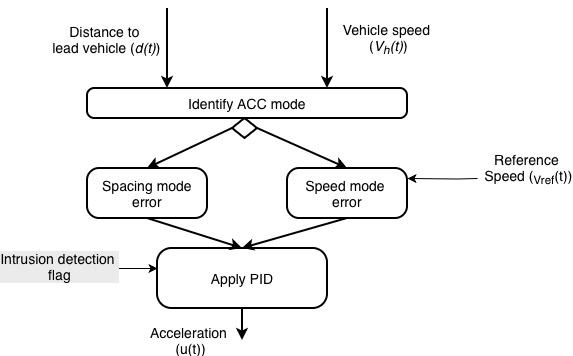}
\caption{Block diagram of ACC with IDS. The model augments the conventional ACC controller with an intrusion flag and a confidence value provided by the IDS.}
\label{fig:BlockDiagramACCWithIDS}
\end{figure}

\begin{equation}\label{eq:PIDIDSCruise}
\begin{split}
u(t) = S(t)\cdot(-a) \;+\;
(1-S(t))\cdot\Big[
&K_p^v \big(v_{\text{ref}} - \hat v_h(t)\big) \\
&+ K_i^v \int_0^t \big(v_{\text{ref}} - \hat v_h(t)\big)\, dt\\
&+ K_d^v \big(\dot{v}_{\text{ref}} - \dot{v}_h(t)\big)
\Big].
\end{split}
\end{equation}

We use a continuous-time PID controller and discretize it with sampling period $\Delta t$. Note that $a \ge 3.4\,\text{m/s}^2$~\cite{greenbook}.

\subsection{Assessing the Effectiveness of ACC Augmented with IDS}

\begin{assumption}\label{as:4}
The IDS employs a machine-learning model trained on CAN-bus data to identify violations involving the host-vehicle speed with accuracy $1.0$. The IDS does not rely on the distance between the host and lead vehicles. Specifically, the IDS flags the following condition as a security violation:
\[
z(t+\Delta t)>
\frac{v_{\text{thr}}(t+\Delta t) - (1 - K_{t+1})\hat v_h(t \mid t)}{K_{t+1}}
\;\;\Rightarrow\;\; S(t)=1.
\]
\end{assumption}

In general, an \ac{IDS} for the in-vehicle network can flag speed values that fall below or exceed the stated threshold as potential intrusions. However, speed values above the threshold should be detected and reported via plausibility checks (e.g., unusually large or rapid increases in the reported speed).

\begin{assumption}\label{as:5}\textbf{Bounded IDS detection delay.} 
The IDS detects attacks within at most $N$ time steps, i.e.,
\[
N\Delta t \le h.
\]
\end{assumption}

\begin{assumption}\label{as:6}\textbf{IDS latch.}
Once the IDS detects an attack at time $\bar t$, the override remains active thereafter:
\[
\forall k\in\mathbb Z_{\ge0},\qquad S(\bar t+k\Delta t)=1.
\]

\end{assumption}

\begin{assumption}\label{as:7}
At the time the ACC--IDS is engaged (initial condition),
\[
d(0) > 0,
\qquad
d(0) \ge d_{\text{safe}}(0).
\]
\end{assumption}

\begin{assumption}\label{as:10}
The lead-vehicle speed is all the time positive.
\[
\forall t \;\; v_{l}(t) \ge 0.
\]
\end{assumption}

\begin{assumption}\label{as:11}
Detection-window host-speed upper bound (worst case).

Letting
\[
M := \frac{\bar t-t^*}{\Delta t}\in \mathbb{Z}_{\ge 0},
\]
there exists a scalar $\bar v$ such that for all $j=1,\dots,M$,
\[
v_h(t^*+j\Delta t)\le \bar v,
\qquad\text{and}\qquad
\bar v \le v_h(t^*).
\]
\end{assumption}

\begin{assumption}\label{as:12}
At the attack start time $t^*$,
\[
d(t^*)\ge d_{\mathrm{safe}}(t^*).
\]
\end{assumption}

Although an ML-based \ac{IDS} is a software component, its latency can be bounded given \ac{CAN} message-rate limits. Assumption~\ref{as:5} is chosen (as a sufficient condition for the theorem) to ensure that the \ac{IDS} detection delay does not exceed the “reaction time” captured by the first term of $d_{\text{safe}}$. Several \ac{IDS} designs proposed in the literature satisfy Assumptions~\ref{as:5} and~\ref{as:6}, including~\cite{JOS2024}. Assumption~\ref{as:4} is not fully realistic; however, accuracies in the range of 0.97–0.99 are commonly reported in the literature, including in our previous work~\cite{9490207}. We report in Section~\ref{sec:simaccuracy} the impact of IDS accuracy on the effectiveness of the \ac{ACC-IDS} controller.

\begin{figure*}[bth]
  \centering
  \begin{subfigure}[b]{0.48\linewidth}
    \centering
    \includegraphics[width=\linewidth]{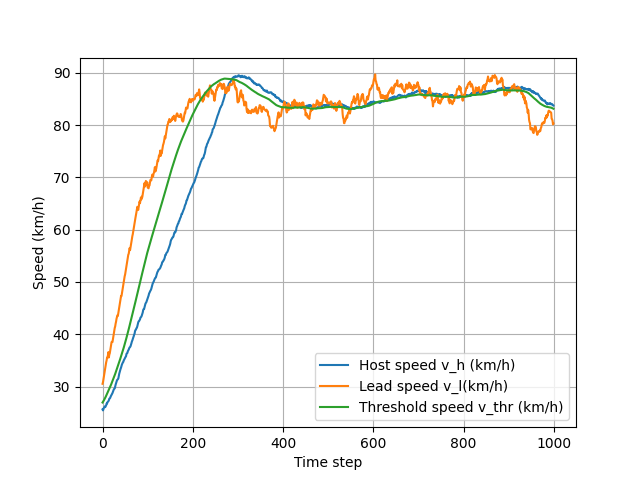}
    \caption{Visualization of the host vehicle speed when using Kalman Filter.}
    \label{fig:speedwithKF}
  \end{subfigure}
  ~
    \begin{subfigure}[b]{0.48\linewidth}
    \centering
    \includegraphics[width=\linewidth]{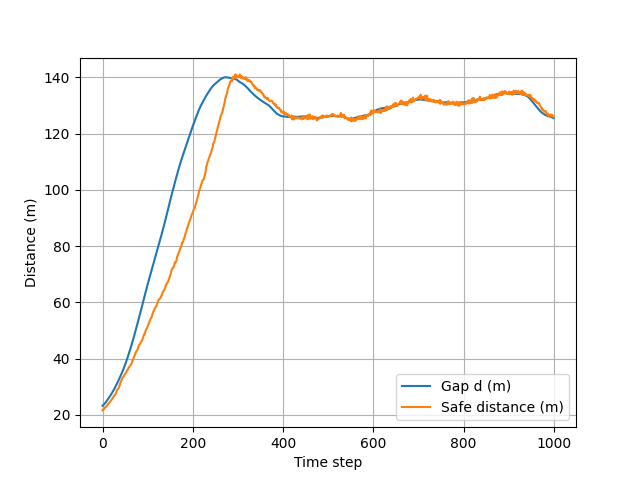}
    \caption{Visualization of the safe distance and gap between the host vehicle speed and lead vehicle when using Kalman Filter.}
    \label{fig:distancewithKF}
  \end{subfigure}
  \caption{Simulation of Adaptive Cruise Control of a host vehicle and lead vehicle where the host vehicle uses Kalman Filter (KF) for the speed readings.}
  \label{fig:SimulationewithKF}
\end{figure*}

\begin{theorem}[IDS braking-distance assures collision avoidance]\label{thm:AccAIDSafe}
Consider the discrete-time longitudinal model with sampling period $\Delta t$:
\begin{align}
v_h(t+\Delta t) &= v_h(t) + u(t)\Delta t, \label{eq:vh_update_thm}\\
d(t+\Delta t) &= d(t) + \bigl(v_\ell(t+\Delta t)-v_h(t+\Delta t)\bigr)\Delta t. \label{eq:d_update_thm}
\end{align}
Define the standard ACC safe-distance decomposition
\begin{align}
d_{\mathrm{safe}}(t) &= d_{\mathrm{react}}(t)+d_{\mathrm{brake}}(t), \label{eq:safe_decomp_thm}\\
d_{\mathrm{react}}(t) &= h\,v_h(t), \\
d_{\mathrm{brake}}(t) &= \frac{v_h(t)^2}{2a}.
\end{align}
Let $t^*$ denote the attack start time and let $\bar t$ denote the (first) IDS detection time.
Assume the ACC--IDS controller enforces emergency braking whenever $S(t)=1$, i.e.,
\[
S(t)=1 \ \Rightarrow\ u(t)=-a.
\]

Suppose that Assumptions~\ref{as:1}-\ref{as:12} hold. Then:

\textup{(A)} \textbf{Braking-distance safety at detection.}
\begin{equation}
d(\bar t)\ \ge\ d_{\mathrm{brake}}(\bar t)=\frac{v_h(\bar t)^2}{2a}.
\label{eq:brake_at_detect_thm}
\end{equation}

\textup{(B)} \textbf{Braking-distance safety for all later times (hence collision avoidance).}
For every integer $k\ge 0$,
\begin{equation}
d(\bar t+k\Delta t)\ \ge\ \frac{v_h(\bar t+k\Delta t)^2}{2a}\ \ge\ 0.
\label{eq:invariant_brake_thm}
\end{equation}
In particular, the inter-vehicle gap never becomes negative, so a collision is avoided.
\end{theorem}

{\bf Proof}
\textbf{Step 1: Detection delay consumes (at most) the reaction-distance term.}
For any sampling instant $s\in[t^*,\,\bar t-\Delta t]$, \eqref{eq:d_update_thm} gives
\[
d(s+\Delta t)=d(s)+\bigl(v_\ell(s+\Delta t)-v_h(s+\Delta t)\bigr)\Delta t.
\]
Using Assumption~9 ($v_\ell(\cdot)\ge 0$) yields the one-step lower bound
\[
d(s+\Delta t)\ge d(s)-v_h(s+\Delta t)\Delta t.
\]
Summing over the $M:=\frac{\bar t-t^*}{\Delta t}\in\mathbb Z_{\ge 0}$ steps from $t^*$ to $\bar t$ yields
\[
d(\bar t)\ge d(t^*)-\sum_{j=1}^{M} v_h(t^*+j\Delta t)\Delta t. \tag{34}
\]
Applying Assumption~10 gives
\[
d(\bar t)\ge d(t^*)-\bar v\,M\Delta t \ \ge\ d(t^*)-\bar v\,N\Delta t, \tag{35}
\]
since $M\le N$ from Assumption~5. Using $N\Delta t\le h$ yields
\[
d(\bar t)\ge d(t^*)-h\bar v. \tag{36}
\]
By Assumption~10, $\bar v\le v_h(t^*)$, hence $h\bar v\le h v_h(t^*)=d_{\mathrm{react}}(t^*)$, and therefore
\[
d(\bar t)\ge d(t^*)-d_{\mathrm{react}}(t^*). \tag{37}
\]
By Assumption~11, $d(t^*)\ge d_{\mathrm{safe}}(t^*)=d_{\mathrm{react}}(t^*)+d_{\mathrm{brake}}(t^*)$.
Combining with (37) yields
\[
d(\bar t)\ge d_{\mathrm{brake}}(t^*)=\frac{v_h(t^*)^2}{2a}. \tag{38}
\]
Finally, Assumption~10 implies $v_h(\bar t)\le \bar v\le v_h(t^*)$, hence
$d_{\mathrm{brake}}(\bar t)\le d_{\mathrm{brake}}(t^*)$, so (38) implies \eqref{eq:brake_at_detect_thm}.
This proves (A).

\medskip
\textbf{Step 2: After detection, braking-distance safety is invariant.}
By Assumption~6 (IDS latch), $S(\bar t+k\Delta t)=1$ for all integers $k\ge 0$,
hence the controller enforces $u(t)=-a$ at all sampling instants $t=\bar t+k\Delta t$, i.e.
\[
v_h(t+\Delta t)=v_h(t)-a\Delta t.
\]
Define the braking margin $\Phi(t):=d(t)-\frac{v_h(t)^2}{2a}$.
Using \eqref{eq:d_update_thm} and $v_h(t+\Delta t)=v_h(t)-a\Delta t$,
\begin{align*}
\Phi(t+\Delta t)
&= d(t)+\bigl(v_\ell(t+\Delta t)-v_h(t+\Delta t)\bigr)\Delta t \\
&- \frac{(v_h(t)-a\Delta t)^2}{2a} \\
&= \Phi(t)+v_\ell(t+\Delta t)\Delta t+\frac{a\Delta t^2}{2}.
\end{align*}
By Assumption~\ref{as:10}, $v_\ell(t+\Delta t)\ge 0$, so $\Phi(t+\Delta t)\ge \Phi(t)$ for all $t\ge \bar t$.
From (A), $\Phi(\bar t)\ge 0$, hence $\Phi(\bar t+k\Delta t)\ge 0$ for all $k\ge 0$,
which is equivalent to \eqref{eq:invariant_brake_thm}. Since $\frac{v_h(\cdot)^2}{2a}\ge 0$,
we also have $d(\bar t+k\Delta t)\ge 0$ for all $k\ge 0$, so a collision is avoided.

\begin{figure*}
  \centering
  \begin{subfigure}[b]{0.48\linewidth}
    \centering
    \includegraphics[width=\linewidth]{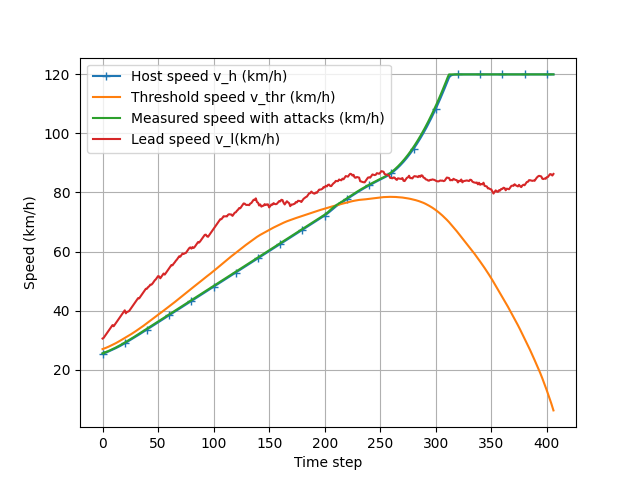}
    \caption{Visualization of the host vehicle speed when using attack injection.}
    \label{fig:speedwithattacks}
  \end{subfigure}
  ~
    \begin{subfigure}[b]{0.48\linewidth}
    \centering
    \includegraphics[width=\linewidth]{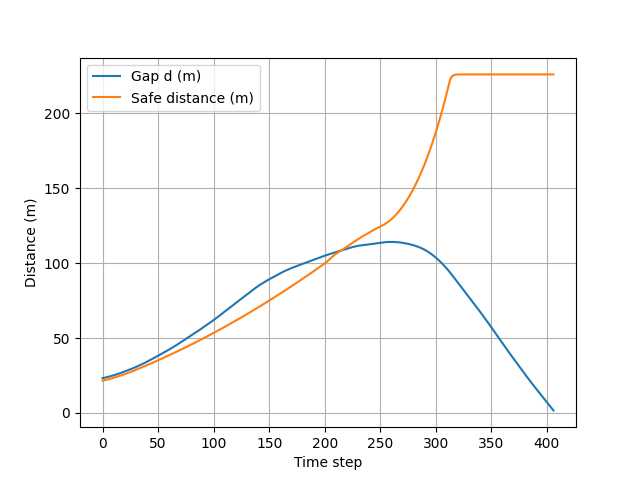}
    \caption{Visualization of the gap between the actual distance and required distance between the host and lead vehicle~\ref{lm:AccSpVilation}}
    \label{fig:distancewithAttacks}
  \end{subfigure}
  \caption{Simulation of Adaptive Cruise Control of a host vehicle and lead vehicle where the host vehicle uses Kalman Filter (KF) for the speed readings in presence of speed injection attacks.}
  \label{fig:SimulationwithAttacks}
\end{figure*}

%% file: Sections/Simulation.tex
This section presents simulation results for \ac{ACC} under three conditions: (i) \ac{ACC} with a \ac{KF}, (ii) \ac{ACC} with a \ac{KF} under speed-injection attacks, and (iii) \ac{ACC} with a \ac{KF} augmented by an \ac{IDS}.

\subsection{Simulation Setup and Data Generation}

We generate simulation data using a discrete-time longitudinal driving simulator that models the interaction between a host vehicle equipped with \ac{ACC} and a lead vehicle traveling in the same lane. The simulator at a fixed sampling period and computes vehicle speed, acceleration, and inter-vehicle distance at each time step. The lead vehicle follows a cruise-control strategy with occasional random braking events, while the host vehicle applies an \ac{ACC} controller that switches between speed tracking (based on speed error) and spacing control (based on distance error).

Each simulation starts from safe initial conditions: the initial inter-vehicle gap is set to $1.1$ times the computed safe following distance. The host vehicle speed measurements are processed by a \ac{KF}, which estimates the vehicle state from measurements corrupted by noise and, under attack, by adversarial manipulation. We model the attack as a sustained sequence of constant positive biases added to the reported speed over a fixed time interval. An \ac{IDS} module is integrated alongside the \ac{KF} to detect abnormal measurement behavior and enforce a fail-safe response.

During each run, the simulator records time-series data including true and estimated speeds, manipuvoids collisionlated measurements, threshold speeds, inter-vehicle gaps, and computed safe distances. The code of the simulation is available at~\cite{OtPr2025}.

\subsection{Simulation of Adaptive Cruise Control with Kalman Filter and No Speed-Injection Attacks}

\begin{figure*}
\centering
\begin{subfigure}[b]{0.48\linewidth}
\centering
\includegraphics[width=\linewidth]{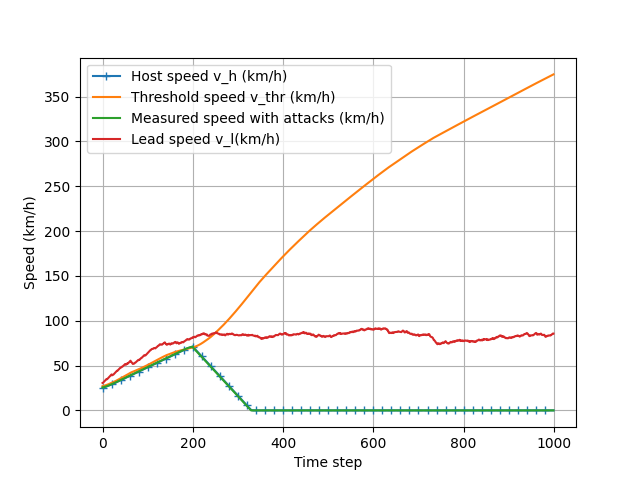}
\caption{Visualization of the host vehicle speed when using \ac{IDS}.}
\label{fig:speedwithIDS}
\end{subfigure}
~
\begin{subfigure}[b]{0.48\linewidth}
\centering
\includegraphics[width=\linewidth]{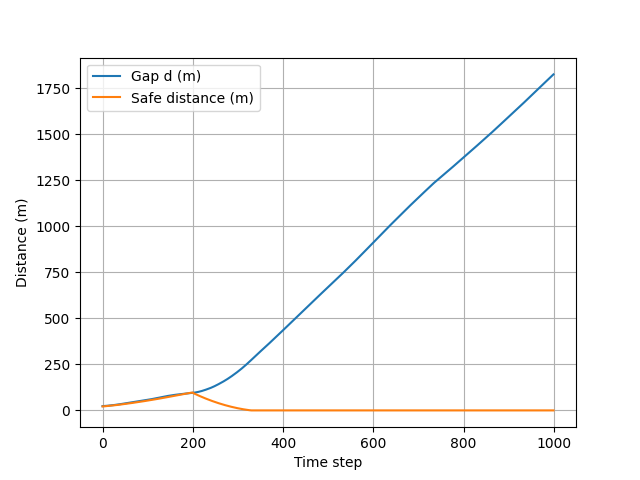}
\caption{Visualization of the gap between the actual distance and required distance between the host and lead vehicle.}
\label{fig:distancewithIDS}
\end{subfigure}
\caption{Simulation of \ac{ACC} augmented with an \ac{IDS}: the host vehicle uses a \ac{KF} for speed estimation under speed-injection attacks.}
\label{fig:SimulationwithIDS}
\end{figure*}

Figure~\ref{fig:SimulationewithKF} shows the simulation results for \ac{ACC} when a \ac{KF} is used under nominal measurement noise (i.e., no attacks). Subfigure~\ref{fig:speedwithKF} plots the host vehicle speed, lead vehicle speed, and the speed safety threshold. Initially, the host speed closely follows the lead vehicle speed while remaining below the threshold, indicating safe operation in speed-tracking mode. When the lead vehicle reaches its maximum speed, it constrains the host’s target speed and the controller transitions to spacing control. We observe occasional brief exceedances of the threshold speed.

Subfigure~\ref{fig:distancewithKF} compares the actual inter-vehicle gap with the required safe distance. When the controller operates in spacing control, the gap and safe distance converge and remain close. Despite the occasional threshold exceedances, \ac{ACC} maintains a safe following distance and avoids collisions.

\subsection{Simulation of Adaptive Cruise Control with Kalman Filter Under Speed-Injection Attacks}

Figure~\ref{fig:SimulationwithAttacks} presents results under a speed-injection attack when only the \ac{KF} is used. Subfigure~\ref{fig:speedwithattacks} shows the host and lead vehicle speeds, the threshold speed, and the manipulated (attacked) speed measurement. After the attack begins, the injected bias causes the measured host speed to increase progressively, eventually reaching the maximum speed. The threshold speed begins to decrease once the host speed exceeds the threshold. The \ac{ACC} uses the manipulated speed to compute both the safe distance and the acceleration commands.

Subfigure~\ref{fig:distancewithAttacks} compares the actual inter-vehicle gap with the safe following distance computed from the manipulated measurement. Under sustained small biases, the \ac{KF} becomes ineffective: the sequence of speed increases initially enlarges the gap, but the gap then rapidly collapses, drops below the required safe distance, and approaches zero, indicating a collision. This divergence between the true gap and the computed safe distance highlights how adversarial speed manipulation can drive unsafe behavior. Overall, the plots show that a \ac{KF} alone is insufficient to prevent collisions when speed measurements are maliciously altered.

\subsection{Simulation of Adaptive Cruise Control Augmented with Intrusion Detection}

Figure~\ref{fig:SimulationwithIDS} illustrates the effectiveness of augmenting \ac{ACC} with an \ac{IDS} during speed spoofing attacks. Subfigure~\ref{fig:speedwithIDS} shows the host and lead vehicle speeds, the manipulated speed measurement, and the threshold speed when the \ac{IDS} is enabled. Although the attacker injects biased measurements that deviate from the true vehicle dynamics, the \ac{IDS} triggers a fail-safe response: it overrides the \ac{ACC} command by suppressing throttle input and enforcing maximum braking, preventing further acceleration.

Subfigure~\ref{fig:distancewithIDS} compares the actual inter-vehicle gap with the required safe following distance. Unlike the attack-only case, the gap increases monotonically and remains well above the safe distance following \ac{IDS} intervention, eliminating collision risk. These results indicate that IDS-assisted control can enforce fail-safe behavior even when Kalman-filter-based estimation becomes unreliable under adversarial sensing.

\subsection{Impact of the Accuracy of the Machine Learning Model Used by the IDS}
\label{sec:simaccuracy}
Across the three scenarios, the figures show results from a representative simulation run of 1000 time steps. Although each plot depicts a single run, we repeated the simulations multiple times and obtained consistent high-level outcomes—most importantly, whether a collision occurred—thereby supporting the theoretical results.

Assuming perfect \ac{IDS} accuracy enables a qualitative assessment of the benefit of augmenting \ac{ACC} with intrusion detection. In practice, however, \ac{IDS} solutions for in-vehicle networks typically achieve high but non-unit accuracy, which makes the guarantee in Theorem~\ref{thm:AccAIDSafe} probabilistic. Importantly, \ac{ACC} failure is generally driven by repeated injections of manipulated speed measurements rather than a single corrupted value; therefore, collisions may still be avoided even when \ac{IDS} accuracy is below 1. 

\begin{figure}[tb]
    \centering
    \includegraphics[width=\linewidth]{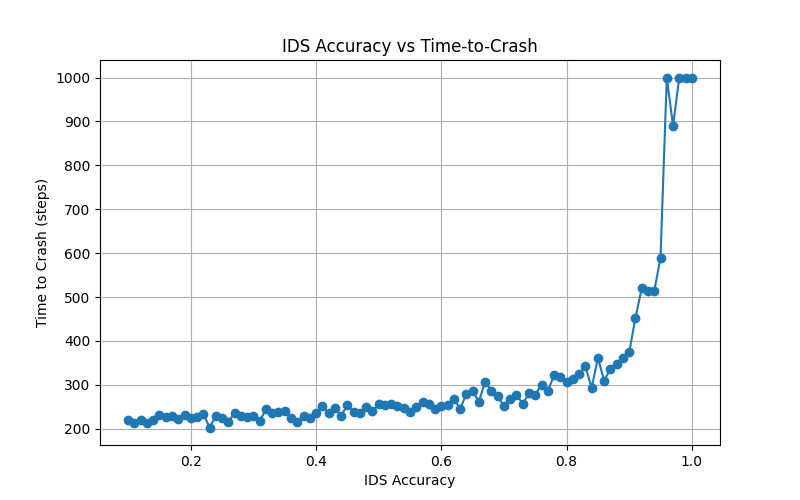}
    \caption{Relationship between IDS accuracy and time to crash (in steps).}
    \label{fig:IDSAccuracyvsCrashes}
\end{figure}

Figure~\ref{fig:IDSAccuracyvsCrashes} illustrates the relationship between IDS accuracy (0.1--1.0) and time to crash (1--1000 steps). Although the plot shows results from a single run per simulation configuration (rather than an average over repeated runs), we observed minimal variation across repeated trials, indicating that the trends are consistent. Overall, the relationship appears approximately parabolic: time to crash increases as IDS accuracy improves, with a pronounced gain in resilience once accuracy exceeds roughly 0.95. Future work will include varying the parameters and types of attacks and braking comfort level, and likely other parameters, on the effectiveness of the solution.

%% file: Sections/Conclusion.tex
We study in this paper spoofed injection attacks against \ac{ACC} and we argue that a \ac{KF} only tolerates injected speed values up to a bounded threshold. Beyond this threshold, the ACC may make unsafe acceleration decisions and crash. The paper (i) derives a next-step speed threshold that separates safe from unsafe evolution of the inter-vehicle gap (Lemma 4.1), (ii) derives a condition on the measured speed that can push the \ac{KF} estimate above this threshold (Lemma 4.2), and (iii) proposes an ACC–IDS architecture in which the IDS triggers enforce braking upon detecting an injection via a modified cruise-mode controller (Eq. 31), along with a safety theorem under several assumptions, including perfect IDS accuracy and bounded detection delay.

We conclude that a \acf{KF} alone is insufficient to mitigate speed fault-injection attacks and cannot reliably prevent such attacks from undermining the \acf{ACC} system’s collision-avoidance capability with respect to the lead vehicle.